%% file: 2024-cdc-trajectorylpcontrol.tex
\documentclass[letterpaper, 10 pt, conference]{ieeeconf}

\IEEEoverridecommandlockouts  
\usepackage{graphicx}

\input{preamble/common}

\title{\LARGE \bf
Spline Trajectory Tracking and Obstacle Avoidance for Mobile Agents via Convex Optimization
}
\author{Akua Dickson, Christos G. Cassandras and Roberto Tron
\thanks{$^{1}$Akua Dickson, Christos Cassandras and Roberto Tron are with the Division of Systems Engineering,
        Boston University, 730 Commonwealth Ave, MA 02215, United States
        {\tt\small\{akuad, cgc, tron\}@bu.edu}}%
}
\begin{document}

\maketitle
\thispagestyle{empty}
\pagestyle{empty}

\begin{abstract}
We propose an output feedback control-based motion planning technique for agents to enable them to converge to a specified polynomial trajectory while imposing a set of safety constraints on our controller to avoid collisions within the free configuration space (polygonal environment).
To achieve this, we \begin{enumerate*}
\item decompose our polygonal environment into different overlapping cells
\item write out our polynomial trajectories as the output of a reference dynamical system with given initial conditions
\item formulate convergence and safety constraints as Linear Matrix Inequalities (LMIs) on our controller using Control Lyapunov Functions (CLFs) and Control Barrier Functions (CBFs)
and \item solve a semi-definite programming (SDP) problem with convergence and safety constraints imposed to synthesize a controller for each convex cell.
\end{enumerate*}
Extensive simulations are included to test our motion planning method under different initial conditions and different reference trajectories.
The synthesized controller is robust to changes in initial conditions and is always safe relative to the boundaries of the polygonal environment.
\end{abstract}

\section{INTRODUCTION}
The primary goal of path planning is to generate a trajectory for a mobile agent between an initial state and a goal state while ensuring collision avoidance at all times; it is a critical cornerstone of modern autonomous systems and robotics.

\myparagraph{Previous work}
Several path-planning algorithms have been proposed in the robotics literature and have been applied to various problems. They share the fundamental goal of finding a feasible collision-free path in a specific environment from an initial state to a goal state\cite{Campbell2020Path}. We briefly review here the most representative approaches and their characteristics.
Potential Field Algorithms are a traditional, computationally cheap method that produces relatively smooth trajectories by attracting the agent toward the goal state, but at the risk of becoming trapped in local minima \cite{victerpaul2017path}.
Alternative algorithms that are \emph{complete} (i.e., they avoid local minima) are \astar{} and \rrtstar{} (Optimized Rapidly Exploring Random Trees) algorithms which successfully avoid local minima and are considered the state of the art in high-dimensional planning spaces since they have good practical performance\cite{victerpaul2017path,Karaman2011Anytime}. Although sample-based planning techniques like \rrtstar{} have good performance and are optimal, they typically do not produce smooth trajectories.
One way to smooth out trajectories is by using direct collocation methods \cite{schulman2014motion} and other trajectory optimization methods.
These methods parameterize the trajectory using a discrete number of points, and then solve the planning problem as a constrained nonlinear optimization problem.
This class of algorithms is well-suited for problems with complex system dynamics and complex constraints, but offers only local convergence guarantees, and it typically needs to be initialized from an initial solution found by \astar{} or \rrtstar{}.

Cell decomposition is another popular path planning algorithm where the specified convex polygonal environment is decomposed into a set of cells \cite{Liu2017Planning}; \astar{} is used to find a feasible sequence of cells, and other techniques are then used to produce paths in each cell. \rrtstar{} algorithms may also be used to build the decomposition of the polygonal environment as seen in \cite{Bahreinian2021Sample-based}.
In this context, and in the context of collocation methods,
it is advantageous to parameterize smooth trajectories using a polynomial basis, such as spline curves\cite{Nguyen2021B-spline,Kano2018Spline,Egerstedt2009Control}. In recent years, splines have been applied in autonomous ground  \cite{Mercy2018Spline,Christian2017Spline-based}, and aerial \cite{Mellinger2011Minimum,Lai2018Optimal,Judd2001Spline} vehicles. The main advantage of this approach is the strong convex hull property of splines: by using a Bernstein polynomial basis, every trajectory will be contained in the convex hull of its control points. Therefore, by choosing the control points that lie completely within a convex collision-free set, we enforce the property that all points in the polynomial trajectory are also guaranteed to be collision-free.
B\'ezier curves are an alternative class of polynomial curves that are defined using a Bernstein polynomial basis and generate smooth polynomial trajectories for path planning \cite{Zheng2020Bezier,Choi2008Path,Zafer2022ANew}. However, B\'ezier curves cannot be implemented for path planning along higher-order polynomial trajectories.

Most of the path planning methods summarized above produce individual paths that need to be tracked using a separate low-level controller. In order to achieve convergence, Control Lyapunov Functions (CLFs) have been employed extensively for nonlinear control of input-affine systems. CLFs are the natural extension of Lyapunov functions to systems with control inputs, but use point-wise optimization to compute a control input $u$ that can asymptotically stabilize the system as desired. Similar concepts are used in Control Barrier Functions (CBFs) \cite{Ames2014Control,Ames2019Control,ames2017control}, which extend barrier functions to ensure safety (specifically, forward-invariance of safe sets) of control-affine systems. 

In the case of unexpected events or disturbances, changes in the goal of the planning problem or infeasible initial conditions, most of the methods reviewed so far would require replanning (i.e., solving the problem almost from scratch).
We build instead upon our previous work \cite{Bahreinian2021Robust,Bahreinian2021Sample-based}, which combines cell-decomposition-based approaches with control synthesis. Here, the trajectories are generated by designing a series of output feedback controllers that take as input the relative displacement of the agent with respect to a set of landmarks, steering the agent through the environment toward a goal location.
This approach is robust to discrepancies and reduces the need for re-planning; however it does not offer an easy way to manipulate directly the final path taken by the agent.

\myparagraph{Paper contributions}
Our approach combines elements of many of the techniques above, namely, control-based cell-decomposition methods, polynomial splines for representing trajectories, and CBFs and CLFs to guarantee convergence and safety over all the points in each cell. These concepts are combined and used to generate constraints for Semi-Definite Program (SDP) optimization problems \cite{boyd1994linear}, which synthesize a sequence of controllers that the agent can use to navigate.

With respect to potential-based methods, we offer completeness (and thus avoid local minima) while still, intuitively, pulling the agent toward a desired path. With respect to \astar{}, \rrtstar-based solutions, direct collocation methods, and basic spline parametrizations, we synthesize low-level, output-feedback controllers directly instead of single reference paths. With respect to cell decomposition methods, we synthesize an output-feedback controller for each cell in the convex polygonal environment that consequently ensures convergence and safety for the entire polygonal environment. By employing CLF-CBF methods, we ensure that the problem is always feasible. Our approach is more efficient than a CBF-based controller, as we do not need to solve an optimization problem online. Moreover, with respect to our own work, we offer convergence to a reference trajectory that can be pre-optimized to improve  performance\footnote{Note that our framework does not preclude the joint optimization of the trajectory and controllers, although this is left as future work}.

Our work consists of two main parts. Firstly, we generate the reference polynomial trajectories by writing these trajectories as the output of a reference dynamical system with 
given initial conditions. The second part involves the design of an output feedback controller in order to track the reference path from an initial state toward a goal state while avoiding collisions and remaining within the given polygonal environment.
We use SDP techniques to solve a sequence of robust convex optimization problems, where the constraints are formulated using CLFs \cite{ames2017control} and CBFs \cite{Springer}
in order to provide stability and safety guarantees respectively, for the given system. The CLF and CBF constraints (the latter after using duality theory) provide, respectively, LMIs and linear inequalities in the optimization problems\cite{Rahmani2021LMI,boyd1994linear}.

In summary, our work generates a reference trajectory from given control points and simultaneously tracks the reference trajectory using an output feedback controller while remaining safe relative to the boundaries of the free configuration space, which is novel, and a more efficient approach to path planning.

\section{PRELIMINARIES}
In this section, we review the key elements required to formulate our research problem. We present the dynamics of the agent, the decomposition of the polygonal environment in which the agent's dynamics evolve and we discuss how the reference polynomial trajectories are generated from the given control points. We then provide a formal problem statement to summarise the research problem we seek to solve. To help the reader in the following sections, we collect in \Cref{table:notation}
a reference list for some of the notation introduced in this work.

\subsection{Dynamical system}
We model the agent using a Linear Time Invariant system of the form
\begin{equation}
\begin{aligned}
\label{eq:sys robot}
\dot{x}&=Ax+Bu\\
y&=Cx
\end{aligned}
\end{equation}
where $x\in\real{d}$ is the state of the system, $y\in\real{d_y}$ is the output, $u\in\real{d_u}$ is the vector of control inputs, and $A,B,C$ are matrices of appropriate dimensions specifying the dynamics of the system.

\myparagraph{Special multidimensional case}
In the following, we will prove additional results for the common case where $y$ can be decomposed as a collection of system, each one with dimension $d_y=1$:
\begin{equation}
\label{eq:robot system dynamics}
\begin{aligned}
A&=\blkdiag(\{A_k\}_{k=1}^{d_y})\\
B&=\blkdiag(\{B_k\}_{k=1}^{d_y})\\
C&=\blkdiag(\{C_k\}_{k=1}^{d_y})\\
x&=\stack(\{x_k\}_{k=1}^{d_y})\\
y&=\stack(\{y_k\}_{k=1}^{d_y})\\
\end{aligned}
\end{equation}
and each $y_k\in\real{}$ is one-dimensional.

\begin{table}[b]
\begin{tabular}{lp{0.75\linewidth}}
$d$ & Dimension of the state $x$\\
$d_u$ & Dimension of the input vector $u$\\
$n_p$ & Order of the polynomial spline.\\
$I_d$ & $d\times d$ identity matrix\\
$0_d$, $0_{d_1,d_2}$ & All-zeros vector/matrix of dimension $d$/$d_1\times d_2$\\
$d_n$ & Dimension of the aggregate state $z$\\
$[x]_i$ & $i$-th element of the vector $x$\\
$x^{(n)}$ & $n$-th derivative of $x$\\
$co(\cdot)$ & convex hull\\
$stack(x,y)$ & $\bmat{x&y}^T$\\
$sym(C)$ & symmetric part of a matrix $C$\\
$blkdiag(x,y)$ & $\bmat{x&0\\0&y}$

\end{tabular}
\caption{Table of notation}
\label{table:notation}
\end{table}

\subsection{Polynomial trajectories}
We assume trajectories of a polynomial form $p(t):[0,1]\to\real{d}$ where

\begin{equation}
\label{eq:polynomial}
p(t)=\sum_{i=0}^{n_p} a_i t^i=\sum_{i=0}^{n_p} P_ib_{i,n_p}(t),
\end{equation}
for $t\in[0,1]$, where 
\begin{equation}
\label{eq:control points}
P=\bmat{P_0&P_1&P_2&.....&P_{n-1}}\in\real{d\times (n_p+1)}
\end{equation}
is a set of $n$ \emph{control points}, and $b_{i,n}$ represent the Bernstein basis polynomials $b_{i,n}(t)=\binom{n}{i}t^i(1-t)^{n-i}\in\real{}$ , $i \in\{0, 1, 2,\ldots ,n\}$, and $\cA=\bmat{a_0&a_1&\ldots& a_{n-1}}$ is a matrix of polynomial coefficients.
Equation \eqref{eq:polynomial} gives two equivalent representations for the reference polynomial trajectories adopted in this paper.
This equivalence is established in the following lemma proved in \cite{Ezhov2021Spline}.
\begin{lemma}
\label{lemma:lemma1}
Polynomial representations with both matrix $P$ and $\cA$ can be written as
\begin{equation}
    \cA=PD
\end{equation}
Given the standard polynomial basis ($\cA$), $D$ is an invertible transformation matrix that derives the Bernstein polynomial basis ($P$) from the standard polynomial basis such that:
\begin{equation}
    \stack(\{b_{i,n}(t)\})=D  \stack(\{t^i\}).
\end{equation}
Then $\cA=PD$.
\end{lemma}
The representation that implements the matrix of polynomial coefficients $\cA$ is vital for the controller synthesis in Section \ref{section:controlsynthesis}.
A cubic spline can be written as:
\begin{multline}
    p(t)=a_3t^3+a_2t^2+a_1t+a_0\\
    =P_3t^3+3P_2(t^2-t^3)+3P_1 (t-2t^2+t^3) +P_0(1-3t+3t^2-t^3)
\end{multline}
hence we have 
\begin{equation}
\begin{aligned}
&a_3=P_3-3P_2+3P_1-P_0,&&a_2=3P_2-6P_1+3P_0,\\
&a_1=3P_1-3P_0,&&a_0=P_0
\end{aligned}
\end{equation}

\subsection{Bounding Bernstein polynomials and their derivatives}
Let $b_n\in\real{n+1}(t)$ be the vector of all Bernstein polynomials of order $n$. Let the matrix $H_n\in\real{(n+1)\times n}$ be defined as 
\begin{equation}
    H_n=n\left(\bmat{-I\\0\transpose}+\bmat{0\transpose\\I}\right),
\end{equation}
where $0\transpose$ denotes a row of all zeros; define also
\begin{equation}
    H_{n,q}=\prod_{m=n}^{n-q+1} H_m.
\end{equation}
Then, we have  
\begin{equation}
    \dot{b}_n(t)=H_nb_{n-1}(t)
\end{equation}
and, for higher-order derivatives of order $q$, 
\begin{equation}
    b_n^{(q)}(t)=H_{n,q} b_{n-q}(t).
\end{equation}

Define the $n$-dimensional probability simplex as
\begin{equation}
    \Delta_n=\{\rho\in\real{n+1}: 0\leq\rho\leq 1, \vct{1}\transpose\rho=1\}.
\end{equation}
Thanks to the partition of unity property of the Bernstein polynomials, the polynomial $p(t)$ can be bounded as 
\begin{equation}
\label{eq:bernsteinpolynomialbounds}
    p(t)\in \{P\rho:\rho\in \Delta_{n_p}\}\; \forall t\in[0,1]
\end{equation}
(this is the same as stating that $p(t)\in\cvxhull(\{P_i\})$).
The derivatives of $p(t)$ can then be bounded as
\begin{equation}
\label{eq:bernsteinpolynomialderivativesbounds}
    p^{(q)}(t)\in \{PH_{n_p,q}\rho:\rho\in\Delta_{n_p-q}\}\; \forall t\in[0,1].
\end{equation}
This important property of polynomial splines i.e. $p(t)\in\cvxhull(\{P_i\})$ for all $t\in[0,1]$ \cite{Piegl1995Nurbs}, implies that once the control points $P$ are in a specified polytope, the polynomial trajectory would always be within the polytope, hence its importance. This allows us to synthesize controllers for each polytope while ensuring that the polynomial trajectories remain within the polytope and avoid collisions with the polytope walls at all times. The polynomial trajectories are represented using the matrix of control points $P$ in order to exploit this property.
\subsection{Polygonal Environment Decomposition}
\label{section:environmentdecomp}
We assume a polygonal environment as our free configuration space for both trajectory generation and controller synthesis. The convex polygonal environment $\cE\in\real{d}$ is decomposed into a finite number of convex cells each given as $\cX$. These convex cells may overlap, but must collectively cover the entire environment $\cE$, i.e.,  $\bigcup_{e} \cX_e = \cE$. Each cell $\cX$ is a polytope that can be represented by the inequality $Ax\leq b$ where $x \in \cX$ such that the matrix $A$ and vector $b$ describe the specified polytope.
Without the overlapping region, the control point shared by two adjacent convex cells would lie on the face that both cells have in common. Therefore, introducing overlapping regions is advantageous since this allows for more flexibility in the position of the shared control point, i.e., the shared control point may lie anywhere in the overlapping polygonal space.

\subsection{Problem statement}
The goal of our work is to design an output feedback controller for each convex cell and ensure that within each cell the agent converges to the specified reference polynomial trajectory while remaining safe with respect to the walls of the polygonal configuration environment. Specifically, we want to design an output feedback control of the form

\begin{equation}\label{eq:robot reference control}
    u=K\bmat{y\\p}
\end{equation}
\begin{equation}
\label{eq:K}
 K= \bmat{K_y&K_p}
\end{equation}
where $K_y$ is the component of the matrix $K$ that corresponds to the system output $y$ and $K_p$ is the component of $K$ that corresponds to the reference trajectory $p$. The tracking control objective is
\begin{equation}
\lim_{t\to\infty} \bigl(p(t)-x(t)\bigr)=0
\end{equation}

\section{PROPOSED SOLUTION}
 Our proposed solution comprises of five integral parts. In the first part, we begin by writing the reference polynomial trajectories as the output of a reference dynamical system with given initial conditions. We then consider both the agent system dynamics and the dynamics of the reference polynomial trajectories obtained in part one as a joint dynamical system in the second part. Furthermore, in the third part, we derive stability constraints in order to ensure that the synthesized controllers steer the agent onto the reference polynomial trajectories. The fourth part details the derivation of safety constraints that enable the agent to avoid collisions with the walls of the polygonal configuration space. We conclude by combining both safety and stability constraints in a controller synthesis problem where we design a controller to track the given reference polynomial trajectory without colliding with the walls.
\subsection{Polynomial trajectories}
In this section, we show that any polynomial trajectory of the form \eqref{eq:polynomial} can be written as the output of an autonomous linear dynamical system where the system matrices are fixed by the order of the polynomial, and the initial conditions determine the overall shape.

We formally define the reference system as
\begin{subequations}\label{eq:sys reference}
    \begin{align}
        \dot{x}_p&=A_px_p,\\
        y_p&=C_p x_p,
    \end{align}
\end{subequations}
where $A_p,C_p$ are defined so that $y_p$ follows a prescribed polynomial trajectory $p(t)$. The initial conditions of the reference dynamical system are important because they determine the overall shape of the reference polynomial trajectories.

\begin{lemma}
\label{lemma:lemma3}
Assume $d_y=1$. Given a polynomial $p(t)$ of the form \eqref{eq:polynomial}, let
\begin{subequations}
\label{Lemma}
\begin{align}
   A_p=\bmat{0_{n_p} & I_{n_p}\\0 & 0_{n_p\transpose}},\label{eq:Ap}\\
   C_p=\bmat{I_{d} & 0_{d\times n_pd}}.\label{eq:Cp}
   \end{align}
\end{subequations}
Then $y_p(t)=p(t)=[x_p]_0(t)$ (the first entry of the state vector) for all $t$ if the initial conditions of the system satisfy
\begin{equation}
    [x_p(0)]_i=a_i i!\textrm{ for all } i,\label{eq:xp initial}
\end{equation}
\end{lemma}
\begin{proof}
The statement $y_p(t)=[x_p]_0$ follows by inspection of \eqref{eq:Cp}.
We divide the proof of the remaining statement $[x_p]_0=p$ in two parts.
In the first part, we give an expression of $[x_p]_0$ in terms of a set of integration constants $C_m$ that need to be determined. In the second part we show that the initial conditions in \eqref{eq:xp initial} determine such integration constants.
For the first part, we actually prove a formula for $[x_p]_i$ for any $i$, proceeding by induction, and beginning from the following ansatz:
\begin{equation}
\label{eq:x_i}
 [x_p]_i = \sum_{k=0}^{n_p-i} \frac{C_{i+k}}{k!}t^k,
\end{equation}
where $C_{i+k}$ are the constants of integration to be determined.
Taking the base case $i = n_p$,  (\ref{eq:x_i}) becomes:
\begin{equation}
  [x_p]_{n_p} = \sum_{k=0}^{n_p-n_p} \frac{C_{n_p+k}}{k!}t^k = C_{n_p}
\end{equation}
For the recursion, from the structure of $A_p$ in \eqref{eq:Ap}, and taking the integral of \eqref{eq:x_i},
$[x_p]_{i-1}$ is given as
\begin{multline}
[x_p]_{i-1} = \int [x_p]_{i}\de t =  \sum_{k'=0}^{n_p-i} \int \frac{C_{i+k'}}{k'!}t^{k'}\de t\\
=\sum_{k'=0}^{n_p-i}\frac{C_{i+k'}}{(k'+1)!}t^{k'+1} + C_{i-1}
\end{multline}
where $C_{i-1}$ is a constant of integration.
Performing the change of variable $k=k'+1$, this becomes
\begin{equation}\label{eq:xi minus}
[x_p]_{i-1} = \sum_{k=1}^{n_p-i+1}\frac{C_{i+k}}{k!}t^k + C_{i-1};
\end{equation}
Including $C_{i-1}$ in the sum for $k=0$, then \eqref{eq:xi minus} is the same as \eqref{eq:x_i}.
This concludes the first part of the proof.

For the second part, we first find the relation between the initial conditions $[x_p]_i(0)$ and the integration constants by evaluating \eqref{eq:x_i} for $t=0$:
\begin{equation}
 [x_p]_i(0) = \sum_{k=0}^{n_p-i} \frac{C_{i+k}}{k!}t^k = C_i
\end{equation}

Then, considering the case $i=0$ of \eqref{eq:x_i}, and using \eqref{eq:xp initial} we have
\begin{equation}
\label{eq:[x]_0}
 [x_p]_0 = \sum_{k=0}^{n_p}\frac{C_{k}}{k!}t^{k}
=\sum_{k=0}^{n_p}a_kt^{k}
\end{equation}
Comparing \eqref{eq:polynomial} to \eqref{eq:[x]_0}, we conclude that $[x_p]_0 =p(t)$.
\end{proof}

\begin{corollary}
\label{corollary:corollary1}
For arbitrary $d_y$, we have
\begin{equation}
\label{eq:polynomial_matrix}
\begin{aligned}
   A_p=\blkdiag(\{A_{p,k}\}_{k=1}^{d_y}),\\
   C_p=\blkdiag(\{C_{p,k}\}),\\
   \end{aligned}
\end{equation}
where each $A_{p,k},C_{p,k}$ is given by Lemma \ref{lemma:lemma1}.
\end{corollary}

\subsection{Combined Dynamics}
Define the aggregate state $z=\bmat{x\\x_p}$, $z\in\real{d_n}$, for the multidimensional case as the stack of the states $x$ of the agent and the states $x_p$ of the polynomial trajectory of the form:
\begin{equation}
z=\stack\bigl(
x,\{x_{p_k},\dot{x}_{p_k},\ldots,x^{(n_p-1)}_{p_k}\}_{k=1}^{n_p}\bigr),
\end{equation}
i.e., we stack the agent's state and then, for each dimension, the corresponding part of the reference state and all its derivatives.

Combining \eqref{eq:sys robot}, \eqref{eq:robot reference control}, and \eqref{eq:sys reference}, the closed-loop dynamics can be written as:
\begin{subequations}
\begin{align}
    \dot{z}&=\bmat{A&0\\0&A_p}z+\bmat{B\\0}u_z\\
    y&=\bmat{C & 0}z
\end{align}
\end{subequations}
where
\begin{equation}
    u_z=\bmat{K_y&K_p}\bmat{C&0\\0&I} z.
\end{equation}

\subsection{Convergence Conditions via Linear Matrix Inequalities}
\label{sec:Controller Synthesis}

Considering the system dynamics of the form \eqref{eq:sys robot}, we define a Lyapunov function as follows:\\
\begin{equation}
    V\left(\bmat{x\\x_p}\right)=\norm{x-C_px_p}^2.
\end{equation}
Using the combined state $z$, the function $V$ can then be rewritten in the form
\begin{equation}
    V=z^TMz,
\end{equation}
where $M$ is given by
\begin{equation}
\label{eq:Lyapunov trajectory matrix}
  M=\bmat{I&-C_p\\-C_p^T&C_p^TC_p}.
\end{equation}
Computing the derivative of the Lyapunov function $V$, we obtain
\begin{multline}
\label{eq:constraint in k}
\begin{aligned}
\dot{V}&=\nabla V^T\dot{z}\\
&=z^T\left((M+M^T)^T \cQ+ \cB K\cC \right)z
\end{aligned}
\end{multline}
where $K=\bmat{K_y&K_p}$, $\cQ=\blkdiag(A,A_p)$ $\cB=\stack(B,0_{n_pd\times d})$, and $\cC=\blkdiag(C,I_{n_p})$. 
\begin{definition}
A function $V$ has \emph{relative degree one} with respect to the dynamics \eqref{eq:sys robot} if $\cL_B V\defeq \nabla V\transpose B\neq 0$ (where $\cL_B$ denotes the vector of Lie derivatives with respect to the fields given by the columns of $B$).
\end{definition}
\begin{assumption}\label{assumption:rel degree}
The Lyapunov function $V$ has relative degree equal to one, i.e., $M\cB\neq 0$.
\end{assumption}

Intuitively, Assumption \ref{assumption:rel degree} implies that $\dot V$ explicitly depends on the control $u$ (and hence on the choice of $K$).
\begin{proposition}
\label{proposition:proposition1}
A sufficient and necessary condition that ensures $\dot{V}\leq 0$ for all $z$ is that there exists $\mu\geq 0$ such that
\begin{equation}
\label{eq:convex optimization S lambda}
    S\preceq -\mu I,
\end{equation}
where:
\begin{equation}
\label{eq:convex optimization S}
S=\symm\left((M+M^T)^T\left(\cQ+\cB K \cC \right)\right)
\end{equation}
\end{proposition}
\begin{proof}
Substituting \eqref{eq:convex optimization S} in \eqref{eq:constraint in k}, the derivative of the Lyapunov function $V$ can be written as
\begin{equation}
\dot{V} = z^TSz.
\end{equation}
It follows that $\dot{V}\leq 0$ for all $z$ if and only if $S\prec-\mu{}I\prec0$,
i.e., if $S$ is negative semidefinite.
\end{proof}

In practice, $\mu$ is a bound on the convergence rate of the $V$ toward zero. For systems that can decompose along different dimensions (i.e., multi-dimensional systems) [\ref{eq:robot system dynamics}], 
we can show that the satisfaction of the convergence constraints for each dimension, when they are considered separately (but with a common convergence rate), still ensures the satisfaction of the convergence constraints for the whole system.

\begin{lemma}
Let $K_{yi}$, $K_{pi}$ be matrices that satisfy the constraints \eqref{eq:convex optimization S lambda} considering only the $i$-th system (i.e., using only $A_{pi},C_{pi},A_i,B_i,C_i$) with a common $\mu$. Let $K$ be the matrix formed as
\begin{equation}
K=[\blkdiag\bigl(\{K_{yi}\}\bigr),\blkdiag\bigl(\{K_{pi})\}\bigr)].
\end{equation}
Then, $K$ satisfies the constraints for the joint system (i.e., using $A_p,C_p,A,B,C$).
\end{lemma}
\begin{proof}
Given the decomposed system dynamics in the form seen in equation \eqref{eq:robot system dynamics} and the reference system dynamics for the polynomial trajectory given in Corollary \ref{corollary:corollary1}, the constraint for the convex optimization problem of a one-dimensional system can be written as:
\begin{equation}\label{eq:convex optimization S lambda i}
S_i \preceq -\mu I_i
\end{equation}
where:
\begin{equation}
S_i=\symm\left((M_i+M_i^T)^T\left(\cQ_i+\cB_i K_i\cC_i\right)\right)
\end{equation}
$A_i, B_i, C_{p_i}, A_{p_i}$ and $C_i$ are the matrices that correspond to the dynamics of the i-th system. $M_i$ is the matrix $M$ in \eqref{eq:Lyapunov trajectory matrix} for the i-th system.
We define $S=\blkdiag(\{S_i\})$ such that: $M=\blkdiag(\{M_i\})$, $A=\blkdiag(\{A_i\})$, $B=\blkdiag(\{B_i\})$, $C=\blkdiag(\{C_i\})$, $A_p=\blkdiag(\{A_{p_i}\})$ and $C_p=\blkdiag(\{C_{p_i}\})$\\
Combining \eqref{eq:convex optimization S lambda i} for all $i$, we have
\begin{equation}
S\preceq -\mu \blkdiag(\{I_i\}),
\end{equation}
which is the same as \eqref{eq:convex optimization S lambda}.
Next, with some algebra, one can see that
\begin{multline}
  S=\symm\Bigl((\blkdiag(\{M_i\})+(\blkdiag(\{M_i\}))^T)^T\bigl( \\ \blkdiag(\blkdiag(\{A_i\}) ,\blkdiag( \{A_{p_i}\})+
      \stack(\blkdiag(\\ \{B_i\}),0)  \times \blkdiag(\{K_i\}) \times \blkdiag(\blkdiag(\{C_i\}),I)\bigr)\Bigr)
\end{multline}
which is the same as $S$\eqref{eq:convex optimization S}.

Hence all the conditions of \Cref{proposition:proposition1} are satisfied, and the claim is proved.
\end{proof}

\begin{example}
Given the two-dimensional first-order integrator dynamics of an agent:
\begin{equation}
\label{eq:two-dimfirstintegratordynamics}
 A = \bmat{0&0\\0&0},\hspace{0.4cm}
B = \bmat{1&0\\0&1}, \hspace{0.4cm}
C = \bmat{1&0\\0&1}
\end{equation}
The gain matrix K for the two-dimensional system with dynamics in \eqref{eq:two-dimfirstintegratordynamics} is of the form:
\begin{equation}
\label{eq:K matrix}
K = \bmat{K_{y_1}&0 &K_{p_1}&0\\0&K_{y_2}&0&K_{p_2}}
\end{equation}
where $K_{y_i} \in\real{1\times d}$ and $K_{p_i} \in\real{1\times n_p}$ for $i=\{1,\cdots,d_y\}$
\end{example}

\subsection{Safety Constraints by Control Barrier Functions}
\label{sec:controlbarrier}
Given $A_{h,i}\in\real{1\times d_n}$ we can define the following candidate Barrier Function of the form:
\begin{equation}
\label{eq:CBF}
h_i(x)=A_{h,i}x+b_{h,i}
\end{equation}
where $i={1,\cdots,s_h}$ such that $s_h$ is the number of walls of the convex polygonal space that limit the agent's state $x$. The matrix $A_{h,i}$ and the vector $b_{h,i}$ describe all the faces of the convex polygonal cell that the agent must not collide with.

\subsubsection{Safety Constraints on $K$}
Considering the agent's dynamics \eqref{eq:sys robot} and a continuously differentiable function $h(x)$ 
defining a forward invariant
safe set $\cX$, the function $h(x)$ is a Control Barrier Function (CBF) if there exists $\alpha \in\real{}$ and control input $u$ such that:
\begin{equation}
\label{eq:CBFConstraint}
\cL_{Ax}h(x)+\cL_Bh(x)u(x,x_p)+\alpha^Th(x)\geq0, \forall x\in\cX, x_p\in\cP
\end{equation}
$\cP$ is the vector of bounds for the Bernstein polynomial \eqref{eq:bernsteinpolynomialbounds} and its derivatives \eqref{eq:bernsteinpolynomialderivativesbounds}. Rewriting the CBF constraint to change the inequality direction, we have:\footnote{This change becomes advantageous for the formulation of the overall control synthesis problem in \Cref{sec:Controller Synthesis}.}
\begin{equation}
\label{eq:CBFFFConstraint1}
\begin{aligned}
-&(\cL_{Ax}h(x)+\cL_Bh(x)u(x,x_p)+\alpha^Th(x))\leq0,\\
&\forall x\in\cX,x_p\in\cP
\end{aligned}
\end{equation}

Note that the constraint \eqref{eq:CBFFFConstraint1} must be satisfied for all $x$ in the cell $\cX$. In other words, the control input at each point in the cell should satisfy the CBF and CLF constraints. We implement this by rewriting the constraint \emph{for all $x$}  using the max formulation:
\begin{equation}
\label{eq:CBFconstraints1}
\left[
\begin{aligned}
\max_x&-(\cL_{Ax}h_i(x)+\cL_{B}h_i(x)KCx+\alpha^T h_i(x))\\&\subjectto
     z\in\cZ
\end{aligned}
\right]
    \leq 0
\end{equation}
where $\cZ = \biggl\{\bmat{x\\x_p}:x\in\cX, x_p \in \cP \biggr\}$.

In practice, we aim to find a controller that satisfies this CBF constraint with some margin. We therefore define a margin $\delta$ to achieve the minimum distance from the obstacles.

Therefore, taking a CBF $h_i(x)$ of the form in \eqref{eq:CBF}
where $A_{h,i}\in\real{1\times d_n}$ and $b_{h,i}$ is a scalar, we can write: 

\begin{equation}
\label{eq:CBFconstraints2}
\left[\begin{aligned}
\max_{x,x_p,\delta}&-A_{h,i}[(A+BK_yC+\alpha)x +BK_px_p]\\ 
s.t.: \\
&A_zx\leq b_z\\
&x_p \in \cP 
\end{aligned}
\right]
\leq \delta+\alpha b_{h,i}\\    
\end{equation}

where the matrix $A_z$ and the vector $b_z$ describe each polygonal convex cell as stated in Section \ref{section:environmentdecomp}.

The above constraint is convex in $K$ (because the point-wise maximum over $x$ ensures that the left-hand side of the inequality is a convex function \cite{boyd1994linear}. However, this form of the constraint is not efficient to implement in an off-the-shelf solver because the controller $K$ (which we would like to design) appears bi-linearly with the variable $x$. Instead, we can take the dual of the maximization problem, thereby obtaining:

\begin{equation}
\label{eq:CBFconstraints3}
\left[\begin{aligned}
\min_{\{\lambda_{b,i}\},\{\gamma_{b,i}\},\delta}&\;\lambda^T_{b,i}b_z + \gamma_{b,i}^T\cP\\ s.t.: \\ 
&A^T_z\lambda_{b,i}=[-A_{h,i}\cW
]^T\\
&\gamma_{b,i} = [-A_{h,i}(BK_p)]^T\\
&\lambda_{b,i}\geq 0\\
&i=\{1,\cdots,s_h\}
\end{aligned}\right]
\leq \delta+ \alpha b_{h,i}    
\end{equation}
where $\cW = A+BK_yC+\alpha$.

Note that $K$ appears linearly in the constraint and $s_h$ denotes the number of faces of $\cX$.
\begin{remark}
By strong duality, once a linear programming problem has an optimal solution, its dual is also guaranteed to have an optimal solution such that both their respective optimal costs are equal \cite{Bertsimas1998Introduction}. Therefore the constraint problems \eqref{eq:CBFconstraints2} and \eqref{eq:CBFconstraints3} are equivalent.
\end{remark}
\subsection{Control Synthesis}
\label{section:controlsynthesis}
In this section, we synthesize the controller by solving a convex optimization problem subject to both stability and safety constraints that were formulated in Sections \ref{sec:Controller Synthesis} and \ref{sec:controlbarrier}.

\begin{equation}
\label{eq:finaloptimization1}
\begin{aligned}
\min_{\mu,S, K}\;& \mu\\
s.t.:\\ &
\mu\leq0\\&S=\symm\left((M+M^T)^T\left(\cQ+\cB K\cC\right)\right)\\
&S\preceq \mu I \\
&\left[\begin{aligned}
\min_{\{\lambda_{b,i}\},\{\gamma_{b,i}\}}&\;\lambda^T_{b,i}b_z + \gamma_{b,i}^T\cP\\ s.t.:\\ 
&A^T_z\lambda_{b,i}=[-A_{h,i}
\cW]^T\\
&\gamma_{b,i} = [-A_{h,i}(BK_p)]^T\\
&\lambda_{b,i}\geq 0\\
&i=\{1,\cdots,s_h\}
\end{aligned}\right]
\leq \delta+ \alpha b_{h,i}  
\end{aligned}
\end{equation}
Furthermore, we simplify the min-min problem to obtain the optimization problem below:
\begin{equation}
\label{eq:finaloptimization}
\begin{aligned}
\min_{\mu,S,K,\delta,\{\lambda_{b,i}\},\{\gamma_{b,i}\}}\;& \mu\\
s.t.:\\ &
\mu\leq0\\&S=\symm\left((M+M^T)^T\left(\cQ+\cB K\cC\right)\right)\\
&S\preceq \mu I \\
& \lambda^T_{b,i}b_z + \gamma_{b,i}^T\cP \leq \delta + \alpha b_{h,i}\\
&A^T_z\lambda_{b,i}=[-A_{h,i}(A+BK_yC+\alpha)]^T\\
&\gamma_{b,i} = [-A_{h,i}(BK_p)]^T\\
&\lambda_{b,i}\geq 0\\
&\delta \geq 0, i=\{1,\cdots,s_h\}\\
\end{aligned}    
\end{equation}

where $K= \bmat{K_y&K_p}$. The objective function and all the constraints in this optimization problem are linear and $\cX$ is a convex set.
\begin{remark}
Note that when writing the constraint above we moved the $\min$ over $\lambda_{b,i}$ together with $S,K,\delta$. Also, note that this is an SDP problem that can be handled with off-the-shelf solvers     
\end{remark}

\subsection{Initialization and switching controllers}
An agent is initialized in a random position within one of the convex cells. In order for an agent to track the entire polynomial trajectory within the polygonal environment, the trajectory is partitioned into segments. Each convex cell contains a segment of the polynomial trajectory, which is generated using the given control points. Each segment of the polynomial trajectory ends at the final control point for that segment, which exists within the overlapping region between two consecutive convex cells. At this control point, which we refer to as the ``switching point" one polynomial segment ends while another begins.

An output feedback controller is synthesized for each polynomial segment and in effect for each convex cell. This controller which we design in \eqref{eq:finaloptimization}, enables an agent to track the polynomial segment without colliding with the walls of the convex cell. Since a controller is designed for each polynomial segment, at each switching point, where we switch between two segments of the polynomial trajectory, we simultaneously switch between the output feedback controllers as well. By tracking each polynomial segment successfully, we inherently track the entire polynomial trajectory as desired.

\section{SIMULATIONS}
\subsection{Simulation Setup and Results}
In order to assess the effectiveness of the proposed path planning algorithm, we run a set of MATLAB simulations. In our simulations, we generate a two-dimensional polygonal environment composed of ten different convex overlapping cells, given by the polygonal cells with colored walls in Fig. \ref{fig:finalimage}. 
All the reference polynomial trajectories have four control points implying that they are all cubic splines. The polygonal environment is 8-shaped with one self-intersection. 
There is one polynomial segment within each  cell that is generated using the given control points and the agents have two-dimensional single-integrator dynamics \eqref{eq:two-dimfirstintegratordynamics}.

Given the control points corresponding to each reference trajectory, we synthesize a controller for each cell by solving the optimization problem we formulated in \eqref{eq:finaloptimization}. We present the result of multiple, randomly-generated initial conditions (marked as 
$\ast$). In order to further optimize the agent's trajectories, we search for the closest point on the polynomial segment to the agent. The agent then tracks the reference polynomial trajectory from that point instead of from the beginning of the polygonal segment. 

We observe from Fig. \ref{fig:finalimage} that the multiple random initializations of the agent within each convex cell all converge to the reference polynomial trajectory while avoiding the walls of the cell. We also observe that the switch between the various controllers takes place within the overlapping region as we intended. Therefore, each simulated trajectory successfully completes its maneuver within the polygonal environment.  

\begin{figure}
    \centering
    \includegraphics[scale=0.23]{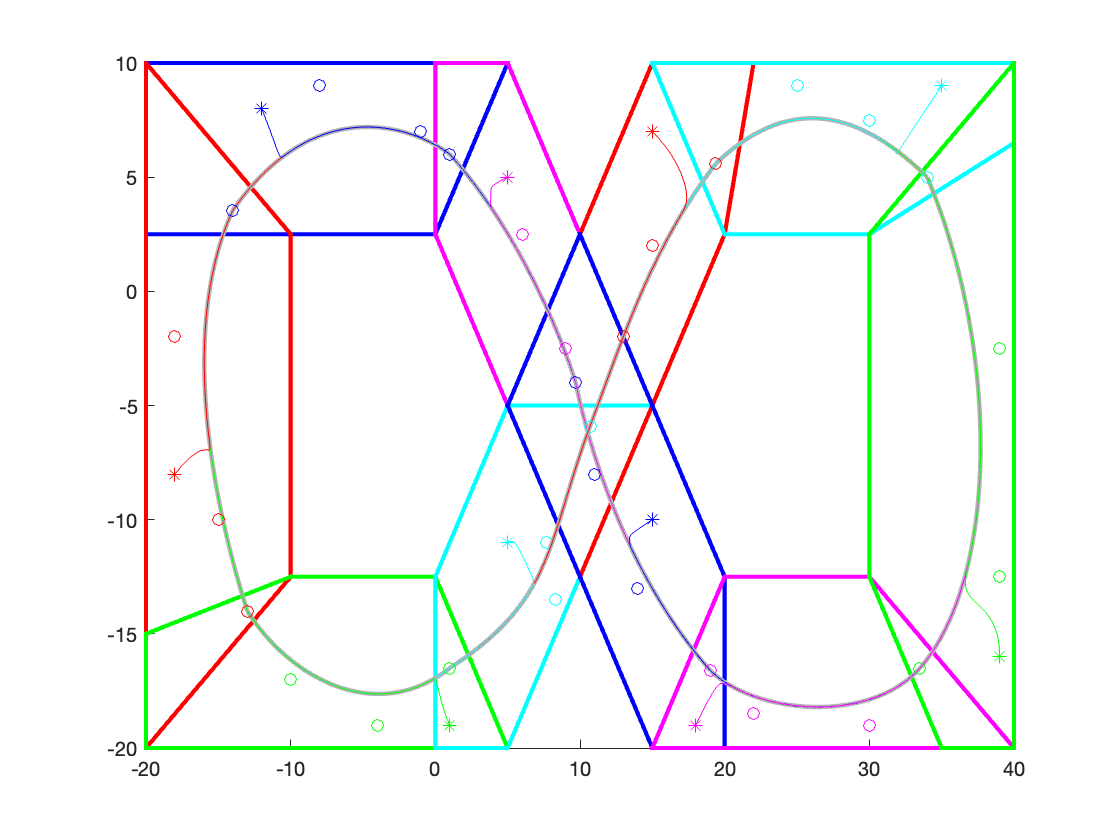}
    \caption{Simulation results without noise. This figure presents multiple initializations of an agent that tracks the predefined polynomial trajectory. Each polygonal cell with colored walls represents a convex cell. Convex cells may overlap. In the middle portion of the environment, several convex cells overlap together. The reference polynomial trajectories corresponding to the convex cells are shown as grey colored lines. The control points defining the segment for each cell have the same colors as the corresponding  cell (except for the final point, which has the color of the following cell). We randomly place an initial position for the agent system; the agent's trajectories are given by the colored lines (although, thanks to our controller, they quickly converge to the gray reference trajectory). The agent trajectory colors correspond to the colored walls of the region in which they are initialized.}
    \label{fig:finalimage}
\end{figure}

\subsection{Simulation Setup With Noise}
In this section, we test the robustness of our path planning algorithm to noise. We introduce Gaussian noise into the dynamics of our agent as seen in Fig. \ref{fig:noise1image}. In particular, at every time step of the simulation, we add a constant random value to the control input. The variance of the noise generated is 0.25.

The controllers implemented in this simulation are exactly the same as before. We observe that our controllers are robust to the presence of Gaussian noise and the agent converges without collisions with the walls of the polygonal configuration space.
\begin{figure}
    \centering
    \includegraphics[scale=0.23]{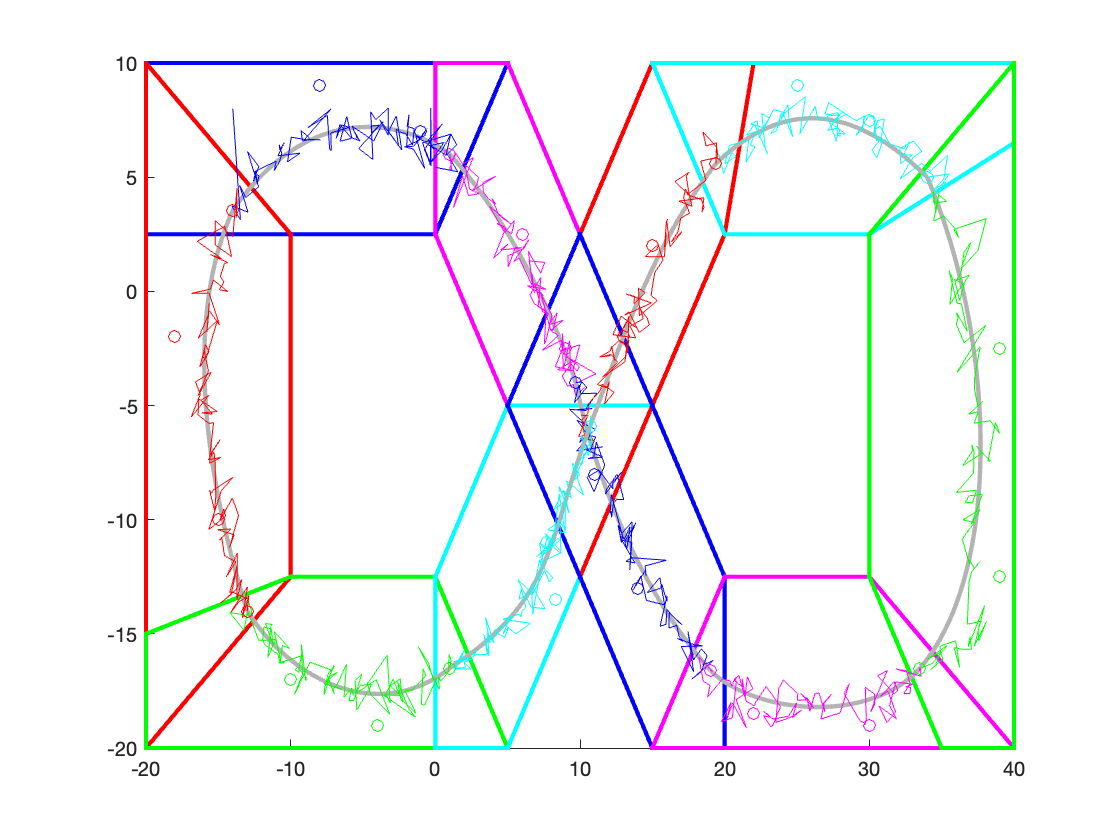}
    \caption{Simulation results with noise. We present results for one random initialization of an agent in the presence of Gaussian noise. The various colored trajectories represent the trajectory of the agent within the polygonal environment. Each colored trajectory corresponds to a convex cell with the same colored walls. Each color represents the segment of the polynomial trajectory under a corresponding controller.}
    \label{fig:noise1image}
\end{figure}

\section{CONCLUSION}
We proposed a novel technique for synthesizing linear output feedback controllers that successfully steer agents onto predefined polynomial trajectories while ensuring safety relative to the walls of the polygonal environment. The resulting controllers are robust to changes in agent initial conditions and noise. 
In the future, we will consider the joint optimization problem of finding optimal reference trajectories (control points) and synthesizing controllers to track said reference trajectories.

\bibliographystyle{biblio/ieee}
\bibliography{main}

\end{document}

%% file: preamble/common.tex
\input{preamble/commonNoTikz}
\input{preamble/graphicsTikz}

\input{preamble/robotics}


%% file: preamble/commonNoTikz.tex
\input{preamble/fixes}
\input{preamble/graphics}
\input{preamble/pagination}
\input{preamble/math}
\input{preamble/operators}
\input{preamble/notation}

\input{preamble/units}

\input{preamble/markupAndCommenting}

\input{preamble/utilities}


%% file: preamble/fixes.tex
\usepackage{fix-cm}
\usepackage{etex}

\usepackage{dblfloatfix}

\usepackage{nag}


%% file: preamble/graphics.tex
\makeatletter
\@ifpackageloaded{xcolor}{}{%
\usepackage[table,x11names,dvipsnames,svgnames]{xcolor}%
}
\makeatother

\usepackage{colortbl}

\usepackage{graphicx}
\usepackage{wrapfig}

\input{preamble/graphicsColors}


%% file: preamble/graphicsColors.tex
\definecolorset{RGB}{lyft}{}{Red,194,39,36;Sunset,202,53,33;Orange,205,68,20;Amber,200,117,42;Yellow,242,169,52;Citron,186,188,44;Lime,112,159,33;Green,56,139,31;Mint,45,118,56;Teal,52,133,135;Cyan,60,132,202;Blue,55,94,248;Indigo,64,13,247;Purple,115,42,248;Pink,176,25,145;Rose,176,32,75}

%% file: preamble/pagination.tex
\usepackage{cite}

\usepackage{microtype}

\usepackage[american]{babel}

\usepackage{array}
\usepackage{multirow}
\usepackage{booktabs}
\usepackage{makecell} 

\newcommand{\myparagraph}[1]{\textbf{\emph{#1}}.}

\ifcsname labelindent\endcsname
\let\labelindent\relax
\fi
\usepackage[inline]{enumitem}

\usepackage{subfig}

\newenvironment{lenumerate}[2][]
{\begin{enumerate}[label=(#2\arabic*),leftmargin=0.2in,itemindent=0.15in,#1]}
{\end{enumerate}}

\setlist*[enumerate,1]{label={\itshape\arabic*)}}

\makeatletter
\newcommand{\paragraphswithstop}{%
\let\copyparagraph\paragraph%
\renewcommand\paragraph[1]{\copyparagraph{##1.}}%
}
\makeatother

\usepackage[framemethod=tikz]{mdframed}

\makeatletter
\def\namedlabel#1#2{\begingroup
  #2%
  \def\@currentlabel{#2}%
  \phantomsection\label{#1}\endgroup
}
\makeatother

\makeatletter
\def\namedlabelphantom#1#2{\begingroup
  \def\@currentlabel{#2}%
  \phantomsection\label{#1}\endgroup
}
\makeatother

\newcommand{\parunskip}{\bgroup\unskip\parfillskip=0pt \par\egroup}


%% file: preamble/operators.tex
\newcommand{\real}[1]{\mathbb{R}^{#1}{}}

\newcommand{\bmat}[1]{\begin{bmatrix}#1\end{bmatrix}}

\newcommand{\transpose}{^\mathrm{T}}

\newcommand{\defeq}{\doteq}

\DeclarePairedDelimiter{\norm}{\lVert}{\rVert}

\newcommand{\de}{\mathrm{d}}

\newcommand{\vct}[1]{\mathbf{#1}}

\DeclareMathOperator{\blkdiag}{blkdiag}
\DeclareMathOperator{\symm}{sym}

\DeclareMathOperator{\stack}{stack}

\DeclareMathOperator{\cvxhull}{co}

\newcommand{\subjectto}{\textrm{subject to}\;}


%% file: preamble/notation.tex
\providecommand{\cA}{\mathcal{A}}
\providecommand{\cB}{\mathcal{B}}
\providecommand{\cC}{\mathcal{C}}

\providecommand{\cE}{\mathcal{E}}

\providecommand{\cL}{\mathcal{L}}

\providecommand{\cP}{\mathcal{P}}
\providecommand{\cQ}{\mathcal{Q}}

\providecommand{\cW}{\mathcal{W}}
\providecommand{\cX}{\mathcal{X}}

\providecommand{\cZ}{\mathcal{Z}}


%% file: preamble/units.tex
\usepackage{units}


%% file: preamble/markupAndCommenting.tex
  \newcommand{\newcolorlabel}[2]{%
  \expandafter\newcommand\csname #1\endcsname[1]{%
    \tikz[baseline]{\node[text=white,fill=#2,anchor=base,text height=1.3ex,text depth=0.1ex,font=\sffamily\bfseries]{##1}}}%
}

\newcommand{\newcommenter}[2]{%
  \expandafter\newcommand\csname #1\endcsname[1]{%
    \fcolorbox{#2}{#2}{\color{white}\textsf{\textbf{#1}}}
    {\color{#2}##1}}%
  \expandafter\newcommand\csname at#1\endcsname{%
    \fcolorbox{#2}{#2}{\color{white}\textsf{\textbf{@#1}}}
    {\color{#2}}}%
  \expandafter\newcommand\csname #1cite\endcsname[1]{%
    \csname #1\endcsname {[##1]}
  }%
  \expandafter\newcommand\csname #1ref\endcsname[1]{%
    \csname #1\endcsname {$\blacktriangleright$##1}
  }%
  \expandafter\newcommand\csname #1hl\endcsname[2]{%
    \colorbox{#2}{\color{white}\textsf{\textbf{#1}}}\sethlcolor{Azure2}\hl{##2}~%
    \expandafter\ifx\csname commentarrow\endcsname\relax$\leftarrow$\else \commentarrow[#2]\fi~%
    {\color{#2}##1}}%
  \expandafter\newcommand\csname #1st\endcsname[2]{%
    \colorbox{#2}{\color{white}\textsf{\textbf{#1}}}\sout{##2}~%
    \expandafter\ifx\csname commentarrow\endcsname\relax$\leftarrow$\else \commentarrow[#2]\fi~%
    {\color{#2}##1}}%
}
\newcommenter{TODO}{DodgerBlue1}
\newcommenter{rtron}{Green3}

\usepackage{comment}

\usepackage{pdfcomment}

\usepackage{soul}

\usepackage[normalem]{ulem}

\usepackage{csquotes}


%% file: preamble/utilities.tex
\usepackage{suffix}

\usepackage{environ}

\makeatletter
\newsavebox{\boxifnotempty}
\newcommand{\displayifnotempty}[3]{\sbox\boxifnotempty{#2}\setbox0=\hbox{\usebox{\boxifnotempty}\unskip}%
  \ifdim\wd0=0pt
  \else
  #1\usebox{\boxifnotempty}#3%
  \fi%
}

\newcommand{\ifempty}[2]{\setbox0=\hbox{#1\unskip}%
  \ifdim\wd0=0pt%
  #2%
  \fi%
}

\newcommand{\ifnotempty}[2]{\setbox0=\hbox{#1\unskip}%
  \ifdim\wd0>0pt%
  #2%
  \fi%
}
\makeatother

\newcommand{\switchifempty}[3]{\sbox\boxifnotempty{#1}\setbox0=\hbox{\usebox{\boxifnotempty}\unskip}%
  \ifdim\wd0=0pt{}%
  #2%
  \else{}%
  #3%
  \usebox{\boxifnotempty}%
  \fi%
}

\makeatletter
\@ifundefined{chapter}{\usepackage{algorithm}}{\usepackage[chapter]{algorithm}}
\makeatother
\usepackage{algorithmicx}
\usepackage{algpseudocode}
\makeatother%

\usepackage{scrlfile}

\makeatletter
\newcommand*\newstoreddef[1]{
  \BeforeClosingMainAux{%
    \immediate\write\@auxout{%
      \string\restoredef{#1}{\csname #1\endcsname}%
    }%
  }%
}
\newcommand*{\restoredef}[2]{
  \expandafter\gdef\csname stored@#1\endcsname{#2}%
}
\newcommand*{\storeddef}[1]{
  \@ifundefined{stored@#1}{0}{\csname stored@#1\endcsname}%
}
\makeatother

\usepackage{pageslts}
\NewEnviron{tee}{\BODY\typeout{Marker Tee [start] ^^J \BODY ^^JMaker Tee [end]}}

\usepackage{cleveref}


%% file: preamble/graphicsTikz.tex
\usepackage{tikz}
\usetikzlibrary{calc}
\usetikzlibrary{matrix}
\usetikzlibrary{chains,scopes}
\usetikzlibrary{shapes.geometric}
\usetikzlibrary{arrows.meta}
\usetikzlibrary{decorations.markings}
\usetikzlibrary{decorations.pathreplacing}
\usetikzlibrary{backgrounds}

\tikzset{
  dim above/.style={to path={\pgfextra{
        \pgfinterruptpath
        \draw[>=latex,|->|] let
        \p1=($(\tikztostart)!1.5em!90:(\tikztotarget)$),
        \p2=($(\tikztotarget)!1.5em!-90:(\tikztostart)$)
        in(\p1) -- (\p2) node[pos=.5,sloped,above]{#1};
        \endpgfinterruptpath
      }
    }
  },
  dim double above/.style={to path={\pgfextra{
        \pgfinterruptpath
        \draw[>=latex,|->|] let
        \p1=($(\tikztostart)!3em!90:(\tikztotarget)$),
        \p2=($(\tikztotarget)!3em!-90:(\tikztostart)$)
        in(\p1) -- (\p2) node[pos=.5,sloped,above]{#1};
        \endpgfinterruptpath
      }
    }
  },
  dim below/.style={to path={\pgfextra{
        \pgfinterruptpath
        \draw[>=latex,|->|] let 
        \p1=($(\tikztostart)!-1em!-90:(\tikztotarget)$),
        \p2=($(\tikztotarget)!-1em!90:(\tikztostart)$)
        in (\p1) -- (\p2) node[pos=.5,sloped,below]{#1};
        \endpgfinterruptpath
      }
    }
  },
}

\tikzset{
    right angle quadrant/.code={
        \pgfmathsetmacro\quadranta{{1,1,-1,-1}[#1-1]}     
        \pgfmathsetmacro\quadrantb{{1,-1,-1,1}[#1-1]}},
    right angle quadrant=1, 
    right angle length/.code={\def\rightanglelength{#1}},   
    right angle length=2ex, 
    right angle symbol/.style n args={3}{
        insert path={
            let \p0 = ($(#1)!(#3)!(#2)$) in     
                let \p1 = ($(\p0)!\quadranta*\rightanglelength!(#3)$), 
                \p2 = ($(\p0)!\quadrantb*\rightanglelength!(#2)$) in 
                let \p3 = ($(\p1)+(\p2)-(\p0)$) in  
            (\p1) -- (\p3) -- (\p2)
        }
    }
}

\newcommand{\pgfextractangle}[3]{%
    \pgfmathanglebetweenpoints{\pgfpointanchor{#2}{center}}
                              {\pgfpointanchor{#3}{center}}
    \global\let#1\pgfmathresult  
}

\usetikzlibrary{shapes.arrows}
\newcommand{\commentarrow}[1][Azure4]{\tikz[baseline=-3pt]{\node[shape border uses incircle, fill=#1,rotate=180,single arrow, inner sep=1pt, minimum size=6pt, single arrow head extend=2pt]{};}}


%% file: preamble/robotics.tex
\tikzset{ax/.style={-latex,line width=2pt}}

\tikzset{camera/.style={fill=Sienna1,fill opacity=0.5},%
image plane/.style={draw=RoyalBlue3,line width=2pt}}

\newcommand{\rrtstar}{$\mathtt{RRT^*}$}
\newcommand{\astar}{$\mathtt{A^*}$}